\documentclass{article}

\usepackage{hyperref}
\usepackage{amsmath, amsthm, amsfonts}

\newtheorem{definition}{Definition}
\newtheorem{theorem}{Theorem}
\newtheorem{lemma}{Lemma}

\newtheorem{proposition}{Proposition}

\begin{document}
	
\title{Composable, Unconditionally Secure Message Authentication without any Secret Key}
\author{Dimiter~Ostrev%
	\thanks{Dimiter Ostrev is with the Interdisciplinary Centre for Security, Reliability and Trust, University of Luxembourg, 6, Avenue de la Fonte, L-4364 Esch-sur-Alzette, Luxembourg, e-mail: dimiter.ostrev@uni.lu}%
}
\maketitle

\begin{abstract}
	We consider a setup in which the channel from Alice to Bob is less noisy than the channel from Eve to Bob. We show that there exist encoding and decoding which accomplish error correction and authentication simultaneously; that is, Bob is able to correctly decode a message coming from Alice and reject a message coming from Eve with high probability. The system does not require any secret key shared between Alice and Bob, provides information theoretic security, and can safely be composed with other protocols in an arbitrary context. 
\end{abstract}

\section{Introduction}

Message authentication allows the receiver to verify that the message comes from the legitimate sender and not from an adversary. Along with secrecy, authentication is one of the most fundamental properties in cryptography. It has direct real world applications, for example in ensuring that the order for a financial transaction comes from somebody authorized to perform the transaction, and not a criminal. Authentication is also used as a primitive in many other cryptographic protocols, for example key exchange protocols, where it serves to protect against man-in-the-middle and impersonation attacks. 

When defining and proving the security of an authentication scheme, we distinguish between computational and unconditional security. In the first case, the definition and proof rely on the assumption that the adversary has limited computational resources, and often also on the conjecture that a certain problem cannot be solved within the specified resource bound. On the other hand, unconditional security makes no assumption on the computational resources available to an adversary; the scheme is guaranteed to be secure against adversaries with unbounded resources. 

Another important aspect of the definition of security is whether it provides composability guarantees or not. It is known that certain definitions of security, although intuitively appealing, fail to guarantee that a cryptographic scheme remains secure in an arbitrary context. One of the known examples is a criterion based on the accessible information used in early security proofs for Quantum Key Distribution. Reference \cite{konig2007small} shows that it is possible for a protocol to satisfy this security criterion, but nevertheless the resulting key cannot be used for one-time pad encryption of a message whose header is known to the adversary. Examples such as this one motivate the introduction of frameworks for composable security such as \cite{canetti2001universally,backes2004general,maurer2011abstract}; protocols proven secure in such a framework are guaranteed to compose safely with other protocols in the framework, and to remain secure in an arbitrary context. 

In this paper, we consider the strongest possible type of message authentication:  we focus on composable, unconditionally secure schemes. It is known that for message authentication to work, Alice and Bob need to have some initial advantage over their adversary Eve; otherwise, Eve can impersonate Alice to Bob and Bob to Alice. What can the initial advantage be? 

Previous research on authentication in information theoretic cryptography has focused on the scenario in which Alice and Bob share randomness that is secret from their adversary Eve.  Information theoretically secure authentication can be achieved using universal$_2$ classes of hash functions \cite{wegman1981new}: Alice and Bob share a secret key $k$ that encodes a particular function $h_k$ from a suitable class of hash functions. To send message $m$ to Bob, Alice computes the tag $t = h_k(m)$ and sends $(m,t)$. To verify that $(m,t)$ comes from Alice, Bob checks that $t = h_k(m)$. Research on this scenario has focused on finding suitable classes of functions for authentication and on proving lower bounds on the secret key size needed for a given level of security; see, for example, \cite{simmons1984authentication,simmons1988survey,stinson1991universal,stinson1994combinatorial,maurer2000authentication}. A variant of the basic scheme for authentication by universal hashing involves recycling part of the key when authenticating multiple messages; this was proposed in \cite[Section 4]{wegman1981new}. Recently, the composable security of authentication by universal hashing both with and without key recycling has been established \cite{portmann2014key}. 

A strong motivation for exploring different possibilities to obtain a composable, unconditionally secure authenticated channel comes from the study of information-theoretically secure key distribution protocols in classical \cite{maurer1993secret,ahlswede1993common} and quantum \cite{bennett1984quantum,ekert1991quantum} cryptography. These key distribution protocols require interaction between the honest participants over an authenticated channel. If Alice and Bob need an initial secret key for authentication, these protocols become key expansion rather than key distribution protocols. The investigation of whether the requirements for authentication can be lowered \cite{maurer2003secretPartI,maurer2003secretPartII,maurer2003secretPartIII,renner2003unconditional,renner2004exact} led to the development of interactive authentication protocols, in which only partially secret and partially correlated strings suffice. In \cite{renner2003unconditional}, an interactive protocol for authentication is proposed that works even if the adversary knows a substantial fraction of the secret key. In \cite{renner2004exact}, it was shown that this interactive authentication protocol, combined with an information reconciliation protocol, can work even in the case when the randomness initially given to Alice and Bob is not perfectly but only partly correlated.

In this paper, we depart from the model of common randomness shared by Alice and Bob. The inspiration for this comes from the work Wyner \cite{wyner1975wire} on the wiretap channel and Csizar and Korner \cite{csiszar1978broadcast} on the broadcast channel with confidential messages. In these papers, it is shown that if the channel from Alice to Bob is less noisy than the channel from Alice to Eve, suitable encoding and decoding exist which accomplish error correction and secrecy simultaneously: Bob can correctly decode Alice's message, but Eve remains ignorant of it. In the present paper, we ask whether a similar phenomenon is possible for authentication instead of secrecy, and we give an affirmative answer. 

Another motivation for the present work is the study of authentication in the context of quantum key distribution. The results we prove in this paper, combined with the analysis of the composition of QKD and authentication \cite{muller2009composability,portmann2014cryptographic}, show that QKD can be performed over insecure classical and quantum channels, between two parties who share no randomness initially, provided that the classical channel between them is less noisy than the channel between them and the adversary. 

As far as the present author is aware, the idea of using an advantage in channel noise for composable, unconditionally secure message authentication has not been explored before. The closest that the present author has been able to find in the cryptography literature is \cite{lai2009authentication}, which considers the problem of running a traditional, key-based authentication protocol over a wiretap channel. Also interesting are a number of methods for authentication used in physical layer security for wireless networks. These methods exploit unique characteristics of the software or hardware of different devices, or unique characteristics of the channel between two locations, to identify legitimate from malicious signals. An overview of these techniques can be found in the surveys \cite[Section VIII-D]{mukherjee2014principles} and \cite{zeng2010non}. 

The multiple access channel from network information theory \cite[Section 15.3]{cover2006elements} is also related to the present paper in that the multiple access channel has many senders and one receiver. However, all the senders and the receiver in the multiple access channel cooperate; they choose their encoding and decoding rules together so as to achieve certain rates of transmission from each sender to the receiver. In our setup, Alice and Bob cooperate, but Eve is malicious. She observes the encoding and decoding rules that Alice and Bob have agreed upon, and tries her best to fool Bob into accepting a message from her as if it is a genuine message from Alice. 

The rest of the paper is structured as follows. In Section \ref{sec:Preliminaries}, we introduce the notation and certain basic results that we will use. In Section \ref{sec:AbstractCryptography}, we introduce Abstract Cryptography, the framework for composable security that we will use. In Section \ref{sec:AuthenticationFromChannelNoise} we formally explain how an authenticated channel can be constructed from an advantage in channel noise, and we prove that the construction is composable and provides information theoretic security. In Section \ref{sec:Extensions} we discuss some extensions of the results from the previous section, and in Section \ref{sec:Conclusion} we conclude the paper and note some possible directions for future work. 

\section{Preliminaries}\label{sec:Preliminaries}

We will often treat the set $\{0,1\}^n$ as a vector space over the field with two elements; thus, for $v,w \in \{0,1\}^n$, $v+w = (v_1+w_1, \dots, v_n+w_n)$ and for a vector $v$ and a subset $S$, $v+S = \{v+w : w \in S\}$. 

By a Bernoulli random variable with parameter $p$ we mean a random variable $X$ such that $Pr(X=0) = 1-p$ and $Pr(X=1)=p$. We will often work with sequences of i.i.d. Bernoulli random variables, where the abbreviation i.i.d. stands for independent, identically distributed. 

The notion of typical sequences plays a central role in information theory: 
\begin{definition}
	A sequence $x=(x_1,x_2, \dots x_n) \in \{0,1\}^n$ is called $\delta$-typical for a sequence $X=(X_1, X_2, \dots, X_n)$ of i.i.d. $Bernoulli(p)$ random variables if \[ |\frac{1}{n} \log Pr(X=x) + h(p) | < \delta\] where $h(p) = -p \log(p) - (1-p) \log (1-p)$ is the binary entropy function.\footnote{All logarithms are taken to base 2. } We denote the set of all $\delta$-typical sequences for $n$ i.i.d. $Bernoulli(p)$ random variables by $T(n,p,\delta)$. 
\end{definition}

An important result in information theory is the Theorem of Typical Sequences \cite[Theorems 3.1.1-3.1.2]{cover2006elements}:
\begin{theorem}\label{thm:TypicalSequences}
	
		 Let $X_1, X_2, \dots$ be a sequence of i.i.d $Bernoulli(p)$ random variables.  Then,  \[ \forall \delta>0, \lim_{n \rightarrow \infty} Pr((X_1, \dots X_n) \in T(n,p,\delta)) = 1 \] In addition, we have the bound \[|T(n,p,\delta)| \leq 2^{n (h(p)+\delta)} \]on the number of $\delta$-typical sequences
	
\end{theorem}

A simple but fruitful model for a noisy communication channel is given by the Binary Symmetric Channel. The Binary Symmetric Channel with parameter $p$ acts on each input bit independently, transmitting it faithfully with probability $(1-p)$ and flipping it with probability $p$. Thus, when the vector $v \in \{0,1\}^n$ is input into this channel, the output is $v+U$ where $U = (U_1, \dots U_n)$ is a vector of i.i.d. $Bernoulli(p)$ random variables. 

\section{Abstract Cryptography}\label{sec:AbstractCryptography}

In this section, we introduce Abstract Cryptography, the framework for composable security that we will use. The general case of Abstract Cryptography was introduced in \cite{maurer2011abstract}; however, for our purposes, it is sufficient to consider the special case for honest Alice and Bob and malicious Eve as developed in \cite{maurer2011constructive}. 

\subsection{An algebra of resources and converters}

By a resource, we mean a system with three interfaces where Alice, Bob and Eve can enter inputs and receive outputs. We will denote resources by calligraphic letters, for example $\mathcal{R}$. It will be convenient to specify the functionality of a resource by giving pseudo-code for it; for example, a channel from Alice to Bob that provides authentication but no secrecy can be described as "on input $m$ from Alice, output $m$ to Bob and Eve. On input $m'$ from Eve, output $\bot$ to Bob." where we use $\bot$ to denote an error message. 

On the set of resources, we have a parallel composition operation, denoted by $\|$, which takes two resources and returns another resource. Thus, $\mathcal{R} \| \mathcal{S}$ is a resource that provides Alice, Bob and Eve with access to the interfaces of both $\mathcal{R}$ and $\mathcal{S}$. 

By converter, we mean a system with an inside and an outside interface, where the inside interface interacts with a resource and the outside interface interacts with a user. If $\alpha$ is a converter, $\mathcal{R}$ is a resource and $i \in \{A,B,E\}$ is an interface, then $\alpha_i \mathcal{R}$ is another resource, where user $i$ has the interface of the converter $\alpha$, and the other two users have their usual interfaces to $\mathcal{R}$. 

\subsection{Distinguishers, distance, construction}

By a distinguisher, we mean a system with four interfaces, three of which connect to the interfaces of a resource, and the fourth one outputs $0$ or $1$. Thus, a distinguisher $\mathcal{D}$ connected to a resource $\mathcal{R}$ is a system that outputs a single bit. 

We use distinguishers to define a notion of distance between resources: 
\begin{definition}
	The distance between two resources $\mathcal{R}, \mathcal{S}$ is \[ d(\mathcal{R}, \mathcal{S}) = \sup_{\mathcal{D}} | Pr(\mathcal{D}\mathcal{R} = 1) - Pr(\mathcal{D}\mathcal{S}=1) |\]
\end{definition}
We take the supremum over all distinguishers $\mathcal{D}$, placing no restriction on their computational resources. This corresponds to choosing to consider unconditional security (in another terminology information theoretic security). 

From the definition, we can prove that $d(\cdot,\cdot)$ has the properties of a pseudo-metric on the set of resources: 
\begin{proposition}
	For all resources $\mathcal{R}, \mathcal{S}, \mathcal{T}$
	\begin{enumerate}
		\item (Identity) $d(\mathcal{R}, \mathcal{R})=0$. 
		\item (Symmetry) $d(\mathcal{R}, \mathcal{S}) = d(\mathcal{S}, \mathcal{R})$.
		\item (Triangle inequality) $d(\mathcal{R},\mathcal{S}) + d(\mathcal{S}, \mathcal{T}) \geq d(\mathcal{R},\mathcal{T})$. 
	\end{enumerate}
\end{proposition}
We can also prove that $d$ has two additional useful properties, which formally capture the intuition "if $\mathcal{R},\mathcal{S}$ are close, then they remain close in an arbitrary context": 
\begin{proposition}
	For all resources $\mathcal{R}, \mathcal{S}, \mathcal{T}$, converters $\alpha$ and interfaces $i \in \{A,B,E\}$
	\begin{enumerate}
		\item (Non-increasing under a converter) \[d(\alpha_i\mathcal{R},\alpha_i\mathcal{S}) \leq d(\mathcal{R},\mathcal{S})\]
		\item (Non-increasing under a resource in parallel) \[d(\mathcal{R}\|\mathcal{T}, \mathcal{S}\|\mathcal{T}) \leq d(\mathcal{R},\mathcal{S})\]
	\end{enumerate}
\end{proposition}
To prove this proposition, observe that a subset of all distinguishers apply the converter $\alpha$ or add the resource $\mathcal{T}$ in parallel.

Before we can proceed to the definition of construction, we need to introduce protocols, filters, and simulators. By a protocol, we mean a pair of converters, one for Alice and one for Bob. By a filter, we mean a converter for Eve's interface of a resource which blocks malicious actions from Eve; we will use symbols such as $\sharp, \flat$ to denote the filters for different resources. By a simulator, we mean a converter for Eve's interface of a resource; the goal of a simulator is to make the interface of one resource appear as the interface of another. 

Now we are ready to define construction. 
\begin{definition}\label{def:Construction}
	We say that a protocol $\pi = (\pi_A, \pi_B)$ constructs resource $\mathcal{S}$ from resource $\mathcal{R}$ within $\epsilon$, denoted $\mathcal{R} \xrightarrow{\pi, \epsilon} \mathcal{S}$, if 
	\begin{enumerate}
		\item ($\epsilon$-close with Eve blocked) $d(\pi_A\pi_B\sharp_E \mathcal{R}, \flat_E \mathcal{S}) < \epsilon$
		\item ($\epsilon$-close with full access for Eve) There exists a simulator $\sigma_E$ such that $d(\pi_A\pi_B\mathcal{R}, \sigma_E \mathcal{S}) < \epsilon$. 
	\end{enumerate}
\end{definition}
The typical interpretation of the definition of construction is the following: $\mathcal{S}$ is the goal, the ideal functionality that Alice and Bob want to achieve. $\mathcal{R}$ is the real resource that they have available. The combination of $\pi$ and $\mathcal{R}$ is required to be indistinguishable from $\mathcal{S}$ in two scenarios: with Eve blocked and with Eve present. 

Since Eve's interfaces to $\mathcal{R}$ and $\mathcal{S}$ may be different, we need to allow for the simulator $\sigma$ in the second condition of the definition. If $\mathcal{S}$ is considered secure, then $\sigma_E \mathcal{S}$ should be considered at least as secure; this is because a subset of all strategies for Eve against $\mathcal{S}$ apply the converter $\sigma$. 

\subsection{General composition theorem}

The notion of construction provides both parallel and sequential composition, as captured in the following theorem \cite[Theorem 1]{maurer2011constructive}: 

\begin{theorem}
	\begin{enumerate}
		\item (Parallel Composition) If $\mathcal{R} \xrightarrow{\pi, \epsilon} \mathcal{S}$ and $\mathcal{R}' \xrightarrow{\pi', \epsilon'} \mathcal{S}'$ then $\mathcal{R} \| \mathcal{R}' \xrightarrow{\pi \| \pi', \epsilon + \epsilon'} \mathcal{S} \| \mathcal{S}'$. 
		\item (Sequential Composition) If $\mathcal{R} \xrightarrow{\pi, \epsilon} \mathcal{S}$ and $\mathcal{S} \xrightarrow{\pi', \epsilon'} \mathcal{T}$ then $\mathcal{R} \xrightarrow{\pi' \pi , \epsilon + \epsilon'} \mathcal{T}$. 
		\item (Identity) For the identity protocol $\mathbf{1} = (\mathbf{1}_A, \mathbf{1}_B)$ and any resource $\mathcal{R}$, $\mathcal{R} \xrightarrow{\mathbf{1}, 0} \mathcal{R}$. 
	\end{enumerate}
\end{theorem}

This theorem captures formally the idea that if an ideal resource can be constructed from a real resource and a protocol, then the construction can safely be used instead of the ideal resource in an arbitrary context. 

\section{Constructing an authenticated channel from an advantage in channel noise}\label{sec:AuthenticationFromChannelNoise}

In this section, we show how Alice and Bob can use an advantage in channel noise to construct an authenticated channel. 

First, we look at the goal: the ideal authenticated channel that Alice and Bob want to construct. The resource $\mathcal{A}^n$ for transmitting $n$-bit authenticated messages from Alice to Bob is defined by the pseudo-code: 
\begin{enumerate}
	\item On input $m \in \{0,1\}^n$ from Alice, output $m$ to Bob and Eve.
	\item On input $m'$ from Eve, output $\bot$ to Bob.
\end{enumerate}
Thus, Bob gets the guarantee: if anything other than $\bot$ is output by the channel, then it must have come from Alice. 

Next, we look at the noisy channel that Alice and Bob have available. Let $0 \leq p < q \leq 1/2$ and consider the resource $\mathcal{N}^n_{p,q}$ defined by the pseudo-code: 
\begin{enumerate}
	\item On input $m \in \{0,1\}^n$ from Alice, draw $U_1, U_2, \dots U_n$ i.i.d. $Bernoulli(p)$ random variables and output $m+U = (m_1 + U_1, \dots, m_n + U_n)$ to Bob. Also output $m$ to Eve. 
	\item On input $m' \in \{0,1\}^n$ from Eve, draw $V_1, \dots, V_n$ i.i.d. $Bernoulli(q)$ random variables and output $m' + V=(m'_1 + V_1, \dots, m'_n + V_n)$ to Bob.  
\end{enumerate}
Thus, $n$-bit messages from Alice go through a binary symmetric channel with parameter $p$, while $n$-bit messages from Eve go through a binary symmetric channel with parameter $q$. 

To construct the ideal from the real resource, Alice and Bob use suitable encoding and decoding of messages. We will denote by $E^n$ Alice's encoding for transmission over $\mathcal{N}^n_{p,q}$, and by $D^n$ Bob's corresponding decoding. Our main result is the following:

\begin{theorem}\label{thm:Main}
	Let $0 \leq p < q \leq 1/2$. Then, for any $r < h(q) - h(p)$, for any $\epsilon>0$ and for all sufficiently large $n$, there exists a protocol $\pi^n = (E^n,D^n)$ such that $\mathcal{N}^n_{p,q} \xrightarrow{\pi^n, \epsilon} \mathcal{A}^{rn}$. 
\end{theorem}

To prove this theorem, we observe that there are two ways that the real system $E^n_A D^n_B \mathcal{N}^n_{p,q}$ can fail: 
\begin{enumerate}
	\item Alice sends a message to Bob, which he decodes incorrectly or rejects. We call this decoding error and denote the maximum probability of it occurring by $p_{de}$. 
	\item Eve sends a message to Bob, which he accepts and decodes. We call this false acceptance and denote the maximum probability of it occurring by $p_{fa}$
\end{enumerate}

Then, in the first part of the proof, we show that there exist suitable encoding for Alice and decoding for Bob such that $p_{de}, p_{fa}$ are both small. This is stated formally in the following proposition, which we prove in subsection \ref{subsec:GoodEncodingAndDecodingExist}: 
\begin{proposition}\label{prop:GoodEncodingAndDecodignExist}
	Let $0\leq p < q \leq 1/2$. Then, for any $r < h(q) - h(p)$, any $\epsilon>0$ and all sufficiently large $n$, there exist encoding and decoding of $rn$ bit messages to $n$ bit codewords such that $p_{de} < \epsilon$ and $p_{fa} < \epsilon$. 
\end{proposition}

In the second part of the proof, we show that if a real system has small probability of decoding error and of false acceptance, then this real system constructs the ideal system in the sense of Definition \ref{def:Construction}. In section \ref{subsec:FromErrorsToConstruction}, we show the following: 

\begin{proposition}\label{prop:FromErrorsToConstruction}
	Let $\pi^n = (E^n, D^n)$ be a protocol encoding $rn$ bit messages into $n$ bit codewords. Suppose the real system $E^n_A D^n_B \mathcal{N}^n_{p,q}$ has probability of decoding error $p_{de}$ and probability of false acceptance $p_{fa}$. Then, 
	\begin{enumerate}
		\item $d(E^n_A D^n_B\sharp_E \mathcal{N}^n_{p,q}, \flat_E \mathcal{A}^{rn}) = p_{de}$.
		\item There is a simulator $\sigma$ such that \[d(E^n_A D^n_B \mathcal{N}^n_{p,q}, \sigma_E \mathcal{A}^{rn}) = \max(p_{de}, p_{fa})\] 
	\end{enumerate}
\end{proposition}

Now, we can complete the proof of Theorem \ref{thm:Main}: it follows immediately from Propositions \ref{prop:GoodEncodingAndDecodignExist} and \ref{prop:FromErrorsToConstruction}. All that is left to do is to prove the two propositions, which we do in the following subsections. 

\subsection{Good encoding and decoding exist}\label{subsec:GoodEncodingAndDecodingExist}

In this section, we show that encoding for Alice and decoding for Bob exist that make the probabilities of decoding error and false acceptance both small, thereby proving Proposition \ref{prop:GoodEncodingAndDecodignExist}. We follow the proof of the noisy channel coding theorem \cite[Chapter 7]{cover2006elements} to bound the probability of decoding error, and perform an additional analysis to bound also the probability of false acceptance. 

The encoding for Alice consists of selecting $2^{rn}$ codewords $\{c_1, \dots, c_{2^{rn}}\} \subset \{0,1\}^n$. The decoding for Bob will be typical sequence decoding: Bob will decode the set of output sequences $c_i + T(n,p,\delta)$ to message $i$. More precisely, Bob's decoding can be described by the pseudo-code "on input $y$, if there is a unique $i$ such that $y \in c_i + T(n,p,\delta)$ then output $i$, otherwise output $\bot$."

Now, given $r < h(q)-h(p)$ and $\epsilon > 0$, we choose $\delta < (h(q)-h(p)-r)/3$ and we use the probabilistic method to show the existence of two codebooks for Alice: a codebook of $2^{rn+1}$ codewords achieving an average probability of decoding error at most $\epsilon/2$, and a codebook of $2^{rn}$ codewords achieving a maximum probability of decoding error at most $\epsilon$. 

We focus on the first codebook. We choose random variables $C_1, C_2, \dots, C_{2^{rn+1}}$ independently, uniformly from $\{0,1\}^n$ and let this be our codebook. Now suppose Alice inputs $C_i$ into the channel, and Bob gets $Y=C_i+U$. By the union bound, the probability of decoding error is then 
\begin{multline*} 
Pr(\textit{decoding error on input } C_i) \\
 \leq Pr( Y \notin C_i + T(n,p,\delta)) + \sum_{j \neq i} Pr( Y \in C_j + T(n,p,\delta)) 
\end{multline*}
The first term goes to zero as $n$ goes to infinity, by the theorem of typical sequences \ref{thm:TypicalSequences}. The second term is bounded by \[ 2^{rn+1} \frac{|T(n,p,\delta)|}{2^{n}} \leq 2 \cdot 2^{-n(1-r-h(p)-\delta)}\] which also goes to zero as $n$ goes to infinity. 

Thus, for a random codebook \[ \frac{1}{2^{rn+1}} \sum_{i=1}^{2^{rn+1}} Pr (\textit{decoding error on input } C_i)  \] goes to zero as $n$ goes to infinity. Therefore, for any $\epsilon$ and for all sufficiently large $n$, there exist particular codebooks $\{c_1, \dots, c_{2^{rn+1}}\}$ such that \[ \frac{1}{2^{rn+1}} \sum_{i=1}^{2^{rn+1}} Pr(\textit{decoding error on input } c_i) < \frac{\epsilon}{2}\] Picking the best $2^{rn}$ codewords of such a codebook, we obtain a codebook of size $2^{rn}$ such that the maximum probability of decoding error is at most $\epsilon$. 

Next, we need to analyze the probability that Bob accepts a message coming form Eve. The set of channel outputs that Bob accepts is \[ S \subseteq \cup_{i=1}^{2^{rn}} (c_i + T(n,p,\delta)) \] What is the probability that Eve's message is corrupted to an output in this set? 

Suppose Eve inputs $z$ into the channel, resulting in output $Y=z+V$ for Bob. Then 
\begin{multline*}
Pr(Y \in S) \leq Pr (V \textit{ is not $\delta$-typical}) + |S| 2^{-n(h(q) - \delta)} \\
\leq Pr (V \textit{ is not $\delta$-typical}) + 2^{rn} 2^{n(h(p) + \delta)} 2^{-n(h(q) - \delta)} 
\end{multline*}
Both of these terms go to zero as $n$ goes to infinity. Thus, for all sufficiently large $n$ the probability of false acceptance will be below $\epsilon$.

\subsection{Construction in the sense of Abstract Cryptography.}\label{subsec:FromErrorsToConstruction}

In the previous subsection, we established that it is possible for a real system to achieve simultaneously low probabilities of decoding error and of false acceptance. In this subsection, we show that these low probabilities imply that the real system constructs the ideal system in the sense of Abstract Cryptography. We will do this by proving Proposition \ref{prop:FromErrorsToConstruction}. 

First, it is helpful to take a step back and develop some general tools for evaluating the distance $d(\cdot,\cdot)$ between resources. Our first lemma shows that we can restrict attention to distinguishers following a deterministic strategy: 
\begin{lemma}\label{lemma:DeterministicDistinguisher}
	Let $\mathcal{R}, \mathcal{S}$ be two resources. Then, for any $\epsilon > 0$, there is a deterministic distinguisher $\mathcal{D}$ such that \[ |Pr(\mathcal{DR}=1) - Pr(\mathcal{DS}=1)| > d(\mathcal{R,S}) - \epsilon \]
\end{lemma}

\begin{proof}
	Let $\mathcal{D}'$ be any distinguisher such that 
	\[ |Pr(\mathcal{D'R}=1) - Pr(\mathcal{D'S}=1)| > d(\mathcal{R,S}) - \epsilon \]
	If $\mathcal{D}'$ is deterministic we are done. Otherwise, $\mathcal{D}'$ is a probabilistic mixture of deterministic distinguishers, and there must exist a deterministic $\mathcal{D}$ in this mixture such that 
	%\begin{multline*} 
	\[|Pr(\mathcal{DR}=1) - Pr(\mathcal{DS}=1)|  
	\geq |Pr(\mathcal{D'R}=1) - Pr(\mathcal{D'S}=1) > d(\mathcal{R,S}) - \epsilon \]
	%\end{multline*}  
\end{proof}

Next, we focus on evaluating the distance between resources that provide no interaction or only one round of interaction. It is known that for resources that provide a single output, the distinguishing advantage is half the $l_1$ distance between the output probability distributions: 
\begin{lemma}\label{lemma:AdvantageForSingleOutput}
	Let $\mathcal{R,S}$ be two resources that take no input and provide an output in some discrete set. Then, using $r,s$ to denote the probability distributions over outputs, we have \[ d(\mathcal{R},\mathcal{S}) = \frac{1}{2} \|r-s\|_1 = \frac{1}{2} \sum_x |r(x)-s(x)| \]
\end{lemma}
\begin{proof}
	Let $\mathcal{D}$ be the distinguisher given by pseudo-code "On input $x$, if $r(x) > s(x)$ output 1, else output 0." Then, 
	%\begin{multline*}
	\[Pr(\mathcal{DR} = 1) - Pr (\mathcal{DS} = 1) 
	 = \sum_{x : r(x) > s(x)} (r(x) - s(x))  = \frac{1}{2} \|r-s\|_1 \]
	%\end{multline*} 
	Now let $\mathcal{D}'$ be any other distinguisher. Without loss of generality, assume $Pr(\mathcal{D'R}=1) \geq Pr(\mathcal{D'S}=1)$ (otherwise flip the output bit of $\mathcal{D}'$). Let $t(x)$ be the probability that $\mathcal{D}'$ outputs 1 on input $x$. Then
	\begin{multline*}
		Pr(\mathcal{DR}=1) - Pr(\mathcal{DS}=1) 
		- Pr(\mathcal{D'R}=1) + Pr(\mathcal{D'S}=1) \\
		= \sum_{x : r(x) > s(x)} (1-t(x)) (r(x)-s(x)) 
		+ \sum_{x : r(x) \leq s(x)} t(x) (s(x)-r(x)) 	\geq 0
	\end{multline*}
\end{proof}

Now we extend this result to resources that take one input and return one output. 
\begin{lemma}\label{lemma:AdvantageForOneInputOneOutput}
	Let $\mathcal{R,S}$ be two resources that take an input in some discrete set and provide an output in some (possibly different) discrete set. Let $r(y|x), s(y|x)$ be the respective conditional probabilities over outputs given inputs. Then, 
	\[  d(\mathcal{R,S}) = \max_x \frac{1}{2} \|r(\cdot|x) - s(\cdot|x)\|_1\]
\end{lemma}
\begin{proof}
	From Lemma \ref{lemma:DeterministicDistinguisher} we know that we can restrict attention to deterministic distinguishers. Now, we consider a deterministic distinguisher between $\mathcal{R}$ and $\mathcal{S}$ whose strategy is to enter input $x$. The distinguisher is now in a position to try to tell the difference between the output distributions $r(\cdot|x)$ and $s(\cdot|x)$; by Lemma \ref{lemma:AdvantageForSingleOutput} we know that the best advantage of such a distinguisher is \[ \frac{1}{2} \|r(\cdot|x) - s(\cdot|x)\|_1 \] To complete the proof, it remains to observe that the best distinguishing advantage between $\mathcal{R}$ and $\mathcal{S}$ is obtained by the deterministic distinguisher that uses the optimal input. 
\end{proof}

Now, we can complete the proof of Proposition \ref{prop:FromErrorsToConstruction}:
\begin{proof}
	First, we show that $d(E^n_AD^n_B\sharp_E \mathcal{N}^n_{p,q} , \flat_E\mathcal{A}^{rn}) = p_{de}$. Both resources take a single input $x \in \{0,1\}^{rn}$ at Alice's interface and return a single output $y \in \{0,1\}^{rn}$ at Bob's interface. The ideal resource always has $y=x$, while the real resource occasionally makes an error in the transmission; thus, from Lemma \ref{lemma:AdvantageForOneInputOneOutput}, we have 
	%\begin{multline*} 
	\[d(E^n_AD^n_B\sharp_E\mathcal{N}^n_{p,q},\flat_E\mathcal{A}^{rn}) 
	= \max_x Pr(E^n_AD^n_B\sharp_E\mathcal{N}^n_{p,q} \text{ makes error on input } x) = p_{de} \]
	%\end{multline*}
	
	Next, we consider the second part of Proposition \ref{prop:FromErrorsToConstruction}. First, we have to choose a suitable simulator. When Alice inputs a message $x \in \{0,1\}^{rn}$ to the real resource $E^n_AD^n_B\mathcal{N}^n_{p,q}$, the codeword $c_x$ comes out uncorrupted at Eve's interface. On the other hand, when Alice inputs $x$ to the ideal resource $\mathcal{A}^{rn}$, $x$ itself appears at Eve's interface. Therefore, we want $\sigma$ to take $x$ and convert it to the corresponding codeword $c_x$. Further, the real resource $E^n_AD^n_B\mathcal{N}^n_{p,q}$ expects inputs of size $n$ at Eve's interface, while the ideal resource $\mathcal{A}^{rn}$ expects inputs of size $rn$. Therefore, the simulator $\sigma$ has to convert Eve's inputs of size $n$ into inputs of size $rn$. Since $\mathcal{A}^{rn}$ outputs an error to Bob on any input from Eve, it does not matter how $\sigma$ maps $\{0,1\}^n$ to $\{0,1\}^{rn}$; thus, we can assume for simplicity that $\sigma$ maps any $n$ bit input from Eve to a sequence of $rn$ zeros. To summarize, we choose the simulator $\sigma$ given by the pseudo-code: "On input $x$ at the inside interface, output $c_x$ at the outside interface. On input $z$ at the outside interface, output $rn$ zeros at the inside interface."
	
	Now, we have to evaluate $d(E^n_AD^n_B\mathcal{N}^n_{p,q}, \sigma_E \mathcal{A}^{rn})$. From the point of view of a distinguisher, both the real and the ideal resources are single input single output devices: the inputs $(Alice-In, Eve-In)$ are of the form $(x, "no-input")$ or $("no-input", z)$ and the outputs $(Eve-Out, Bob-Out)$ are of the form $(c_x, y)$ or $("no-output", y)$. Thus, Lemma \ref{lemma:AdvantageForOneInputOneOutput} applies. If the distinguisher chooses an input of the form $(x, "no-input")$, then his maximum advantage is the probability that the real system makes decoding error on input $x$ from Alice. If the distinguisher chooses an input of the form $("no-input", z)$, then his maximum advantage is the probability that the real system does not output an error to Bob. Thus, we obtain \[ d(E^n_AD^n_B\mathcal{N}^n_{p,q}, \sigma_E \mathcal{A}^{rn}) = \max (p_{de},p_{fa}) \] as needed. 
\end{proof}

\section{Extensions}\label{sec:Extensions}

In this section we consider some extensions of the results of the previous section. First, we consider an extension to more general models of a noisy channel. Then, we consider the possibility of proving a converse result.  Next, we consider an extension that allows the adversary to block messages from Alice to Bob. Finally, we consider the computational efficiency of encoding and decoding. 

\subsection{More general models of a noisy channel}

Let $\mathbb{X},\mathbb{Y}, \mathbb{Z}$ be finite alphabets for Alice's input, Bob's output and Eve's input respectively. Let $P(\cdot | \cdot)$ and $Q( \cdot | \cdot)$ be two sets of conditional probabilities and consider the real resource $\mathcal{N}^n_{P,Q}$ for transmitting $n$-symbol words given by the pseudo-code: 
\begin{enumerate}
	\item On input $x= (x_1, \dots, x_n)$ from Alice, output $Y=(Y_1, \dots, Y_n)$ to Bob, where $Y_i$ is drawn independently according to the distribution $P(\cdot| x_i)$. Also output $x$ to Eve. 
	\item On input $z=(z_1, \dots, z_n)$ from Eve, output $Y=(Y_1, \dots, Y_n)$ to Bob, where $Y_i$ is drawn independently according to the distribution $Q(\cdot|z_i)$. 
\end{enumerate}
Thus, Alice's messages pass through a discrete memoryless channel with transition probabilities $P$ and Eve's messages pass through a discrete memoryless channel with transition probabilities $Q$. 

Using essentially the same argument as the proof of Theorem \ref{thm:Main} in Section \ref{sec:AuthenticationFromChannelNoise} we obtain: 
\begin{theorem}\label{thm:GeneralDMC}
	For every 
	\begin{equation}\label{eq:RatesForWhichTheProofApplies}
	r < \sup_{P_X} \Big(\min \{I_P(X;Y), \min_z H_Q(Y|Z=z) -H_P(Y|X) \}\Big)  
	\end{equation}
	for every $\epsilon > 0$ and for all sufficiently large $n$, there exist a protocol
	 $\pi^n = (E^n, D^n)$ such that $\mathcal{N}_{P,Q}^n \xrightarrow{\pi, \epsilon} \mathcal{A}^{rn}$. 
\end{theorem}

In equation \eqref{eq:RatesForWhichTheProofApplies}, $I(\cdot;\cdot)$ denotes the mutual information, $H(\cdot|\cdot)$ denotes the conditional Shannon entropy, and the subscript $P$ or $Q$ denotes the probability mass function which is used to compute the corresponding entropic quantities. The supremum is taken over all probability mass functions $P_X$ on $\mathbb{X}$, where each choice of $P_X$, combined with the transition probabilities $P(\cdot|\cdot)$ induces a joint probability mass function $P_{XY}$ on $\mathbb{X} \times \mathbb{Y}$. 

For the case when both $P$ and $Q$ are weakly symmetric \cite[Section 7.2]{cover2006elements} (i.e. the vectors $P(\cdot|x)$ for different $x$ are permutations of each other and the sums $\sum_x P(y|x)$ are the same for all $y$, and similarly for $Q$), the right hand side of equation \eqref{eq:RatesForWhichTheProofApplies} simplifies to an expression with a nice intuitive interpretation: 
\begin{multline*}
\sup_{P_X} \Big( \min \{ I_P(X;Y), \min_z H_Q(Y|Z=z) - H_P(Y|X) \} \Big) \\
= H_Q (Y |Z=z) - H_P (Y|X=x) \\
= (\log |\mathbb{Y}| - H_P(Y|X=x)) - (\log |\mathbb{Y}| - H_Q(Y|Z=z)) 
= C_{A \rightarrow B} - C_{E \rightarrow B}
\end{multline*}
Thus, if both channels are weakly symmetric, Alice can transmit information to Bob at any rate up to the difference between the capacity of the channel from Alice to Bob and the capacity of the channel from Eve to Bob. 

We proceed to prove Theorem \ref{thm:GeneralDMC}. Again, we look at the two cases of decoding error and false acceptance. Proposition \ref{prop:FromErrorsToConstruction} from Section \ref{sec:AuthenticationFromChannelNoise} carries over to this setting as well, because its proof does not rely on the size of the alphabets at the three terminals. What remains to be done is to show that low probabilities of decoding error and false acceptance are simultaneously achievable. We have the following: 
\begin{proposition}\label{prop:GoodEncodingDecodingGeneralCase}
	For any \[ r < \sup_{P_X} \Big( \min \{I_P(X;Y), \min_z H_Q(Y|Z=z) - H_P(Y|X)\} \Big) \] any $\epsilon>0$ and all sufficiently large $n$, there exists encoding and decoding of $rn$ bit messages into $n$ symbol codewords such that $p_{de}<\epsilon$ and $p_{fa} < \epsilon$. 
\end{proposition}

\begin{proof}
We need the notion of joint typicality \cite[Section 7.6]{cover2006elements}:
\begin{definition}
	Let $(X,Y)$ be a pair of random variables taking values in $\mathbb{X} \times \mathbb{Y}$ with joint probability mass function $P$. Let $(X_1,Y_1), (X_2,Y_2), \dots$ be a sequence of i.i.d. pairs, each pair having the same distribution as $(X,Y)$. Let $(X^n,Y^n)=(X_1\dots X_n, Y_1\dots Y_n)$. An element $(x^n,y^n) = (x_1\dots x_n,y_1\dots y_n) \in \mathbb{X}^n \times \mathbb{Y}^n$ is jointly $\delta$-typical if
	\begin{align*}
	|\frac{1}{n} \log Pr((X^n,Y^n) = (x^n,y^n)) + H(X,Y)| &< \delta \\
	|\frac{1}{n} \log Pr(X^n=x^n) + H(X)| &< \delta \\
	|\frac{1}{n} \log Pr(Y^n=y^n)+ H(Y)| &< \delta 
	\end{align*}
	The set of all jointly $\delta$-typical sequences for length $n$ and probability mass function $P$ is denoted $JT(n,P,\delta)$
\end{definition}

\begin{theorem}\label{thm:JoitnlyTypicalSequences}
		 In the setup from the definition above, we have \[\forall \delta > 0 \;\; \lim_{n \rightarrow \infty} Pr((X^n,Y^n) \in JT(n, P, \delta)) = 1\] Moreover, if $\tilde{X}^n, \tilde{Y}^n$ are independent and have the same marginals as $X^n,Y^n$, then \[Pr((\tilde{X}^n,\tilde{Y}^n) \in JT(n,P,\delta)) \leq 2^{-n(I(X;Y) - 3 \delta)} \]
	
\end{theorem}
Now, we can proceed to prove Proposition \ref{prop:GoodEncodingDecodingGeneralCase}. As in Section \ref{sec:AuthenticationFromChannelNoise}, we follow the proof of the noisy channel coding theorem \cite[Chapter 7]{cover2006elements} to bound the probability of decoding error, and perform an additional analysis to bound also the probability of false acceptance.  Let $P_X$ and $\delta$ be such that \[r + 3 \delta < \min\{I_P(X;Y), \min_z H_Q(Y|Z=z) - H_P(Y|X)\} \]

Alice chooses codewords $C_1, \dots C_{2^{rn+1}}$ at random, with each symbol of each codeword being independent with probability mass function $P_X$. Bob uses jointly-typical decoding: "On input $y^n$, if there is a unique $i$ such that $(C_i, y^n) \in JT(n,P,\delta)$ then decode to $i$, otherwise output $\bot$." 

If Alice inputs $C_i$ into the channel and Bob gets output $Y^n$, then Bob's probability of decoding error is 
\begin{multline*} 
	Pr(\text{Decoding error on input } C_i) \\
	\leq Pr((C_i,Y^n) \notin JT(n,P,\delta)) 
	+ \sum_{j \neq i} Pr((C_j,Y^n) \in JT(n,P,\delta)) 
\end{multline*}
and both terms go to zero as $n$ goes to infinity, by Theorem \ref{thm:JoitnlyTypicalSequences} and the choice of $r, \delta, P_X$. 

Thus, for any $\epsilon > 0$ and all sufficiently large $n$, there exist particular codebooks $\{c_1, \dots c_{2^{rn+1}}\}$ such that \[ \frac{1}{2^{rn+1}} \sum_i Pr(\text{Decoding error on input } c_i) < \frac{\epsilon}{2} \] Picking the best $2^{rn}$ codewords of such a codebook, we obtain a codebook of size $2^{rn}$ and maximum probability of decoding error at most $\epsilon$. 

Next, we need to bound the probability that Bob accepts a message coming from Eve. Let $S_i \subset \mathbb{Y}^n$ be the set of channel outputs that Bob decodes to $i$. We will bound the number of elements of $S_i$: using the definition of joint typicality we get 
%\begin{multline*}
\[2^{-n(H(X)-\delta)} \geq P_{X^n}(c_i) 
\geq \sum_{y^n \in S_i} P_{X^n Y^n}(c_i,y^n) \geq |S_i| 2^{-n(H(X,Y)+\delta)} \]
%\end{multline*} 
so $|S_i| \leq 2^{n(H(Y|X)+2 \delta)}$. 

Now suppose that Eve inputs $z^n$ in the channel and Bob gets output $Y^n$. What is the probability that Bob doesn't decode to $\bot$? 
\begin{multline*}
Pr(Y^n \in \cup_i^{2^{rn}} S_i) 
\leq Pr(Y^n \text{ is not $\delta$-typical}) + 2^{-\sum_{i=1}^n H_Q(Y|Z=z_i) + n \delta} \sum_{i=1}^{2^{rn}} |S_i| \\
\leq Pr(Y^n  \text{ is not $\delta$-typical}) 
+ 2^{-n(\min_z H_Q(Y|Z=z) - H_P(Y|X) - r - 3 \delta)}  
\end{multline*} 
and both terms go to zero as $n$ goes to infinity.\footnote{Note that in bounding the probability that $Y^n$ is not $\delta$-typical, we have used an extension of the theorem of typical sequences to handle the case of a sequence of random variables that are independent but not necessarily identically distributed; this extension has the same proof: Chebyshev's Inequality $\Rightarrow$ Law of Large Numbers $\Rightarrow$ Theorem of Typical Sequences. } This completes the proof of Proposition \ref{prop:GoodEncodingDecodingGeneralCase}. 
\end{proof}

Now, we can also complete the proof of Theorem \ref{thm:GeneralDMC}: it follows immediately from Propositions \ref{prop:FromErrorsToConstruction} and \ref{prop:GoodEncodingDecodingGeneralCase}. 

\subsection{A converse result?}

A natural question is whether one can prove a converse result; that is, whether one can prove that if Alice and Bob attempt to transmit information at a rate \[ r > \sup_{P_X} \Big( \min \{ I_P(X,Y), \min_z H_Q(Y|Z=z) -H_P(Y|X) \} \Big)\] bits per channel use, then they must necessarily sacrifice either error correction or authentication. We give an example showing that this is not the case. 

Let the alphabet for Alice be $\{0,1\}$, the alphabet for Bob be $\{0,1,2,3\}$, and the alphabet for Eve be $\{0,1\}$. Let the transition probabilities from Alice to Bob be \[P(0|0) = P(1|1) = 1-p, \;\; \;\; \;\; P(0|1) = P(1|0) = p\] all other probabilities being zero. Thus, the channel from Alice to Bob is a binary symmetric channel with parameter $p$, that only uses the first two symbols of Bob's alphabet. Let the transition probabilities from Eve to Bob be \[ Q(2|0) = Q(3|1) = 1 \] all other probabilities being zero. Thus, the channel from Eve to Bob is a perfect binary channel that uses only the second two symbols of Bob's alphabet. Then, the upper bound from equation \eqref{eq:RatesForWhichTheProofApplies} is $-h(p) < 0$. Nevertheless, it is clear that Alice and Bob can transmit at any rate up to the capacity $1-h(p)$ of the binary symmetric channel between them and can achieve both authentication and error correction. Indeed, Bob can tell that a message comes from Eve by the presence of output symbols $2,3$ from the channel. 

This example shows that the upper bound on the rate given by equation \eqref{eq:RatesForWhichTheProofApplies} is not a fundamental limit but is an artifact of the particular proof technique used. It also shows that it is possible to simultaneously achieve error correction and authentication in certain cases where the channel from Alice to Bob is \emph{more} noisy than the channel from Eve to Bob. 

\subsection{Adversaries that can block messages}

Certain treatments of authenticated channels, for example \cite{portmann2014key}, allow the adversary to block Alice's messages from reaching Bob for both the real and the ideal resource. We can model this by adding the following line to the pseudo-code of both the real resource $\mathcal{N}^n_{P,Q}$ and the ideal resource $\mathcal{A}^n$: 
\begin{enumerate}
	\item[0.] On input $b \in \{0,1\}$ at a (separate) Eve interface, if $b=0$ then do not output anything to Bob in line 1.  
\end{enumerate}

This extra option for the adversary Eve necessitates a small modification in the proof of Proposition \ref{prop:FromErrorsToConstruction}: the filters $\sharp_E, \flat_E$ have to always input $b=1$ to their respective resources, the simulator $\sigma$ has to convey the bit $b$ from the outside to the inside interface, and the distinguisher $\mathcal{D}$ has to consider inputs $(Alice-In, Eve-In)$ of the form $(x,z,b)$ where $x$ is a string or $"no-input"$, $z$ is a string or $"no-input"$, $b \in \{0,1\}$, and if $b=1$ then at least one of $x,z$ has to be $"no-input"$. 

\subsection{Efficient encoding and decoding}

In this subsection, we return to the Binary Symmetric Channel model from Section \ref{sec:AuthenticationFromChannelNoise}. At first sight, Theorem \ref{thm:Main} looks like an existential result: it states the existence of good encoding and decoding, but does not give an explicit construction, neither does it specify the required computational resources for good encoding and decoding. 

However, if we look closely at the proof, we see that it depends only on the following: the set of all channel outputs that Bob accepts is too small from the point of view of Eve, so that an input from Eve is unlikely to be corrupted into this set. Thus, we can take any class of error correcting codes with efficient encoding and decoding, for example low density parity check codes \cite{gallager1962low,spielman1995computationally,luby2001improved}, and within that class we can choose a code with the number of codewords and the radius of the hamming balls decoded to each codeword as required for the proof of Theorem \ref{thm:Main}. 

\section{Conclusion and future work}\label{sec:Conclusion}

We have shown that if the channel from Alice to Bob is less noisy than the channel from Eve to Bob, then Alice and Bob can accomplish error correction and message authentication simultaneously. The intuition behind the result is that for long sequences, there is a subset $S$ of the channel outputs for Bob such that $S$ is large when measured by the probability that a codeword from Alice is corrupted into it, and $S$ is also small when measured by the probability that any input from Eve is corrupted into it. 

To ensure seamless integration of the authentication scheme proposed here with other cryptographic protocols, we have proved it provides composable, information theoretic security using the Abstract Cryptography framework. We have also shown that error correcting codes with efficient encoding and decoding can be used, as long as the set $S$ of outputs that Bob accepts is small from the point of view of Eve. 

The present paper raises a number of interesting questions that can be the subject of future work; we list some of them here. First, what is the set of all rates $r$ such that Alice and Bob can transmit information at rate $r$ bits per channel use and achieve both error correction and authentication? In the present paper, we have shown that rates up to a certain bound are always achievable, but have also given an example where a rate higher than the bound is possible. Thus, the complete characterization of the achievable rates is still not known. Second, would allowing two way communication and interaction between Alice and Bob give further possibilities, as was the case for secrecy in the wiretap channel \cite{maurer1991perfect}, and for authentication in the shared randomness model \cite{renner2003unconditional,renner2004exact}? Third, is it possible to combine the coding for the broadcast channel and for the authentication channel to achieve error correction, authentication and secrecy simultaneously? 

\section*{Acknowledgments}

This work was supported by the Luxembourg National Research Fund (CORE project AToMS).

\end{document}